\documentclass[11pt]{amsart}

\usepackage{amsmath,amssymb,latexsym,amscd,amsthm} 
\usepackage[numbers]{natbib}
\usepackage{graphicx}
\usepackage{fullpage}

\newcommand{\f}{\hat{f}}
\newcommand{\eps}{\varepsilon}
\newcommand{\ind}[1]{1_{\{#1\}}}
\newcommand{\G}{\mathcal{G}}

\newtheorem{proposition}{Proposition}
\newtheorem{lemma}{Lemma}
\newtheorem{theorem}{Theorem}
\newtheorem{corollary}{Corollary}
\newtheorem{claim}{Claim}
\theoremstyle{remark}
\newtheorem{remark}{Remark}

\begin{document}

\title{The uniqueness property for networks with several origin-destination pairs}
\author{Fr\'ed\'eric Meunier \and Thomas Pradeau}
\address{Universit\'e Paris Est, CERMICS (ENPC) \\F-77455 Marne-la-Vall\'ee}
\email{frederic.meunier@enpc.fr, thomas.pradeau@enpc.fr}
 
\begin{abstract}
 We consider congestion games on networks with nonatomic users and user-specific costs. We are interested in the uniqueness property defined by Milchtaich [Milchtaich, I. 2005. Topological conditions for uniqueness of equilibrium in networks. Math. Oper. Res. 30 225-244] as the uniqueness of equilibrium flows for all assignments of strictly increasing cost functions. He settled the case with two-terminal networks. As a corollary of his result, it is possible to prove that some other networks have the uniqueness property as well by adding common fictitious origin and destination.
 In the present work, we find a necessary condition for networks with several origin-destination pairs  to have the uniqueness property in terms of excluded minors or subgraphs. As a key result, we characterize completely bidirectional rings for which the uniqueness property holds: it holds precisely for nine networks and those obtained from them by elementary operations. For other bidirectional rings, we exhibit affine cost functions yielding to two distinct equilibrium flows. Related results are also proven. For instance, we characterize networks having the uniqueness property for any choice of origin-destination pairs.

\end{abstract}

\maketitle

\section{Introduction}

In many areas, different users share a common network to travel or to exchange informations or goods. Each user wishes to select a path connecting a certain origin to a certain destination. However, the selection of paths in the network by the users induces congestion on the arcs, leading to an increase of the costs. Taking into account the choices of the other users, each user looks for a path of minimum cost. We expect therefore to reach a Nash equilibrium: each user makes the best reply to the actions chosen by the other users.

This kind of games is studied since the 50's, with the seminal works by \citet{Wa52} and \citet{Be56}. When the users are assumed to be nonatomic -- the effect of a single user is negligible -- equilibrium is known to exist~\citep{Mi00}. Moreover, when the users are affected equally by the congestion on the arcs, the costs supported by the users are the same in all equilibria \citep{AM81}. In the present paper, we are interested in the case when the users may be affected differently by the congestion. In such a case, examples are known for which these costs are not unique. Various conditions have been found that ensure nevertheless uniqueness. For instance, if the user's cost functions attached to the arcs are continuous, strictly increasing, and identical up to additive constants, then we have uniqueness of the equilibrium flows, and thus of the equilibrium costs \citep{AK01}. In 2005, continuing a work initiated by \citet{Mi00} and \citet{Ko04} for networks with parallel routes, 
\citet{Mi05} found a topological characterization of two-terminal networks for which, given any assignment of strictly increasing and continuous cost functions, the flows are the same in all equilibria. Such networks are said to enjoy the {\em uniqueness property}. Similar results with atomic users have been obtained by \citet{OrRoSh93} and \citet{RiSh07}.

The purpose of this paper is to find similar characterizations for networks with more than two terminals. We are able to characterize completely the ring networks having the uniqueness property, whatever the number of terminals is. The main result is that it holds precisely for nine networks and those obtained from them by elementary operations. For other rings, we exhibit affine cost functions yielding to two distinct equilibrium flows. It allows to describe infinite families of graphs for which the uniqueness property does not hold. For instance, there is a family of ring networks such that every network with a minor in this family does not have the uniqueness property.

\section{Preliminaries on graphs}

An {\em undirected graph} is a pair $G=(V,E)$ where $V$ is a finite set of {\em vertices} and $E$ is a family of unordered pairs of vertices called {\em edges}. 
A {\em directed graph}, or {\em digraph} for short, is a pair $D=(V,A)$ where $V$ is a finite set of {\em vertices} and $A$ is a family of ordered pairs of vertices called {\em arcs}. A {\em mixed graph} is a graph having edges and arcs. More formally, it is a triple $M=(V,E,A)$ where $V$ is a finite set of vertices, $E$ is a family of unordered pairs of vertices (edges) and $A$ is a family of ordered pairs of vertices (arcs). Given an undirected graph $G=(V,E)$, we define the {\em directed version} of $G$ as the digraph $D=(V,A)$ obtained by replacing each (undirected) edge in $E$ by two (directed) arcs, one in each direction. An arc of $G$ is understood as an arc of its directed version. In these graphs, {\em loops} -- edges or arcs having identical endpoints -- are not allowed, but pairs of vertices occuring more than once -- {\em parallel edges} or {\em parallel arcs} -- are allowed.

A {\em walk} in a directed graph $D$ is a sequence $$P=(v_0,a_1,v_1,\ldots,a_k,v_k)$$ where $k\geq 0$, $v_0,v_1,\ldots,v_k\in V$, $a_1,\ldots,a_k\in A$, and $a_i=(v_{i-1},v_i)$ for $i=1,\ldots,k$. If all $v_i$ are distinct, the walk is called a {\em path}.
If no confusion may arise, we identify sometimes a path $P$ with the set of its vertices or with the set of its arcs, allowing to use the notation $v\in P$ (resp. $a\in P$) if a vertex $v$ (resp. an arc $a$) occurs in $P$.

An undirected graph $G'=(V',E')$ is a {\em subgraph} of an undirected graph $G=(V,E)$ if $V'\subseteq V$ and $E'\subseteq E$.
An undirected graph $G'$ is a {\em minor} of an undirected graph $G$ if $G'$ is obtained by contracting edges (possibly none) of a subgraph of $G$. {\em Contracting} an edge $uv$ means deleting it and identifying both endpoints $u$ and $v$.
Two undirected graphs are {\em homeomorphic} if they arise from the same undirected graph by subdivision of edges, where a {\em subdivision} of an edge $uv$ consists in introducing a new vertex $w$ and in replacing the edge $uv$ by two new edges $uw$ and $wv$.

The same notions hold for directed graphs and for mixed graphs.

Finally, let $G=(V,E)$ be an undirected graph, and $H=(T,L)$ be a directed graph with $T\subseteq V$, then $G+H$ denotes the mixed graph $(V,E,L)$.

\section{Model}\label{sec:model}

Similarly as in the multiflow theory (see for instance \citet{S03} or \citet{KV00}), we are given a {\em supply graph} $G=(V,E)$ and a {\em demand digraph} $H=(T,L)$ with $T\subseteq V$. The graph $G$ models the (transportation) {\em network}. The arcs of  $H$ model the origin-destination pairs, also called in the sequel the {\em OD-pairs}. $H$ is therefore assumed to be simple, i.e. contains no loops and no multiple edges. 
A {\em route} is an $(o,d)$-path of the directed version of $G$ with $(o,d)\in L$ and is called an $(o,d)$-route. The set of all routes (resp. $(o,d)$-routes) is denoted by $\mathcal{R}$ (resp. $\mathcal{R}_{(o,d)}$). 

The population of {\em users} is modelled as a bounded real interval $I$ endowed with the Lebesgue measure $\lambda$, the {\em population measure}. The set $I$ is partitioned into measurable subsets $I_{(o,d)}$ with $(o,d)\in L$, modelling the users wishing to select an $(o,d)$-route.

For a given pair of supply graph and demand digraph, and a given partition of users, we define a {\em strategy profile} as a measurable mapping $\sigma:I\rightarrow\mathcal{R}$ such that $\sigma(i)\in\mathcal{R}_{(o,d)}$ for all $(o,d)\in L$ and $i\in I_{(o,d)}$. 
For each arc $a\in A$ of the directed version of $G$, the measure of the set of all users $i$ such that $a$ is in $\sigma(i)$ is the {\em flow} on $a$ in $\sigma$ and is denoted $f_a$: $$f_a=\lambda\{i\in I:\,a\in\sigma(i)\}.$$

The cost of each arc $a\in A$ for each user $i\in I$ is given by a nonnegative, continuous, and strictly increasing {\em cost function} $c_a^i:\mathbb{R}_+ \rightarrow \mathbb{R}_+$, such that $i \mapsto c_a^i(x)$ is measurable for all $a\in A$ and $x\in\mathbb{R}_+$. When the flow on $a$ is $f_a$, the cost for user $i$ of traversing $a$ is $c_a^i(f_a)$. For user $i$, the cost of a route $r$ is defined as the sum of the costs of the arcs contained in $r$. A {\em class} is a set of users having the same cost functions on all arcs, but not necessarily sharing the same OD-pair.

The game we are interested in is defined by the supply graph $G$, the demand digraph $H$, the population user set $I$ with its partition, and the cost functions $c_a^i$ for $a\in A$ and $i\in I$. If we forget the graph structure, we get a game for which we use the terminology {\em nonatomic congestion game with user-specific cost functions}, as in \citet{Mi96}.

A strategy profile is a (pure) Nash equilibrium if each route is only chosen by users for whom it is a minimal-cost route. In other words, a strategy profile $\sigma$ is a Nash equilibrium if for each pair $(o,d)\in L$ and each user $i\in I_{(o,d)}$ we have
\begin{equation*}\sum_{a\in\sigma(i)}c_a^i(f_a)=\min_{r\in\mathcal{R}_{(o,d)}}\sum_{a\in r}c_a^i(f_a).
\end{equation*}

Under the conditions stated above on the cost functions, a Nash equilibrium is always known to exist. It can be proven similarly as Theorem 3.1 in \citet{Mi00}, or as noted by \citet{Mi05}, it can be deduced from more general results (Theorem 1 of \citet{Sc70} or Theorems 1 and 2 of \citet{Ra92}). However, such an equilibrium is not necessarily unique, and even the equilibrium flows are not necessarily unique.

\section{Results}

\citet{Mi05} raised the question whether it is possible to characterize networks having the {\em uniqueness property}, i.e. networks for which flows at equilibrium are unique. A pair $(G,H)$ defined as in Section~\ref{sec:model} is said to have the {\em uniqueness property} if, for any partition of $I$ into measurable subsets $I_{(o,d)}$ with $(o,d)\in L$, and for any assignment of (strictly increasing) cost functions, the flow on each arc is the same in all equilibria.

Milchtaich found a positive answer for the two-terminal networks, i.e. when $|L|=1$. More precisely, he gave a (polynomial) characterization of a family of two-terminal undirected graphs such that, for the directed versions of this family and for any assignment of (strictly increasing) cost functions, the flow on each arc is the same in all equilibria. For two-terminal undirected graphs outside this family, he gave explicit cost functions for which equilibria with different flows on some arcs exist.

The objective of this paper is to address the uniqueness property for networks having more than two terminals. We settle the case of ring networks and find a necessary condition for general networks to have the uniqueness property in terms of excluded minors or subgraphs. \\

 In a ring network, each user has exactly two possible strategies. See Figure~\ref{fig:mixed} for an illustration of this kind of supply graph $G$, demand digraph $H$, and mixed graph $G+H$. We prove the following theorem in Section~\ref{sec:proof}.

\begin{theorem}\label{thm:main}
Assume that the supply graph $G$ is a cycle. Then, for any demand digraph $H$, the pair $(G,H)$ has the uniqueness property if and only if each arc of $G$ is contained in at most two routes.
\end{theorem}

Whether such a pair $(G,H)$ of supply graph and demand digraph is such that each arc in contained in at most two routes is obviously polynomially checkable, since we can test each arc one after the other. We will show that it can actually be tested by making only one round trip, in any direction, see Section~\ref{subsec:algo}. More generally, Section~\ref{sec:comb} contains a further discussion on the combinatorial structure of such a pair $(G,H)$. Especially, we prove in Section~\ref{subsec:explicit} that such a pair $(G,H)$ has the uniqueness property if and only if $G+H$ is homeomorphic to a minor of one of nine mixed graphs, see Figures~\ref{fig:min1}--\ref{fig:min4}.
Except for the smallest one, none of the uniqueness properties of these graphs can be derived from the results by Milchtaich, even by adding fictitious vertices as suggested p.235 of his article \citep{Mi05}.

\begin{figure}\begin{center} 
 \includegraphics[scale=0.9]{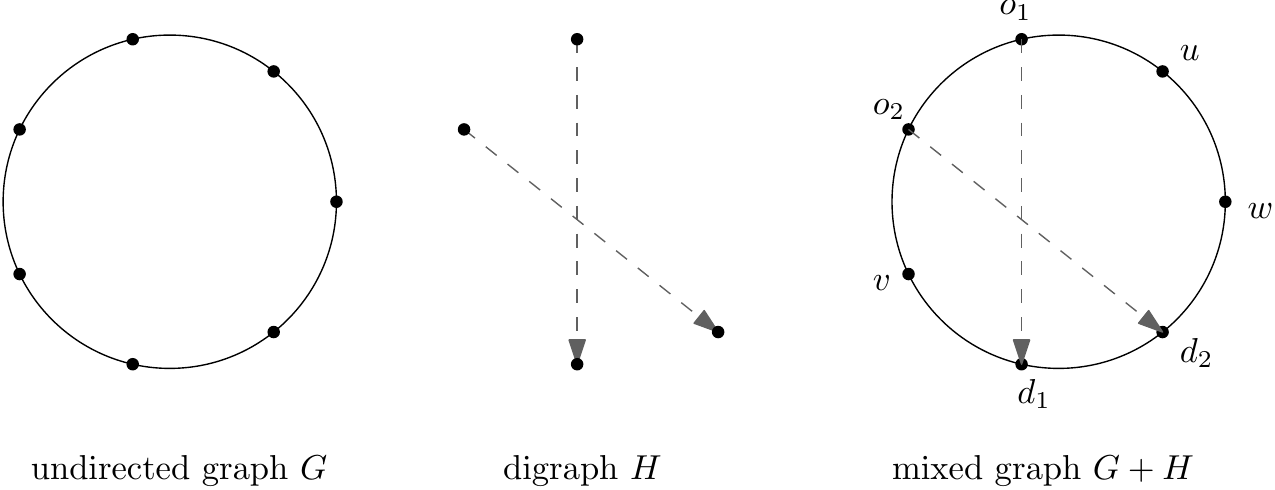}
\caption{Example of a supply graph $G$, a demand digraph $H$, and the mixed graph $G+H$. According to Theorem~\ref{thm:main}, $(G,H)$ has the uniqueness property\label{fig:mixed}}
\end{center}\end{figure} 
 
Furthermore, we find on our track a sufficient condition for congestion games with nonatomic users to have the uniqueness property when each user has exactly two available strategies (Proposition~\ref{prop:more2}).

 

%
Section~\ref{subsec:general} proves a necessary condition for general graphs to have the uniqueness property in terms of excluded minors (Corollary~\ref{cor:minor}). With the help of Theorem~\ref{thm:main}, it allows to describe infinite families of networks not having the uniqueness property. The remaining of Section~\ref{sec:disc} contains complementary results. 
For instance, Section~\ref{subsec:stronguniq} defines and studies a {\em strong uniqueness property} that may hold for general graphs independently of the demand digraph, i.e. of the OD-pairs.

\section{Proof of the characterization in case of a ring}\label{sec:proof}

\subsection{Proof strategy and some preliminary results}\label{subsec:claim}

In this section, we prove Theorem~\ref{thm:main}. The proof works in two steps. The first step, Section~\ref{subsec:uniq}, consists in proving Proposition~\ref{prop:uniq} below stating that, when each arc is contained in at most two routes, then the uniqueness property holds. The second step, Section~\ref{subsec:nonuniq}, consists in exhibiting cost functions for which flows at equilibrium are non-unique for any pair $(G,H)$ with an arc in at least three routes.

From now on, we assume that the cycle $G$ is embedded in the plane. It allows to use an orientation for $G$. Each route is now either positive or negative. The same holds for arcs of $G$: we have positive arcs and negative arcs. 
\begin{claim}\label{claim:sum}
For any $(o,d)\in L$, if $a^+$ and $a^-$ are the two arcs stemming from an edge $e\in E$, then exactly one of $a^+$ and $a^-$ is in an $(o,d)$-route.
\end{claim}
\begin{proof} Indeed, given an $(o,d)\in L$ and an edge $e\in E$, exactly one of the positive and negative $(o,d)$-routes contains $e$.
\end{proof}


For any subset $J\subseteq L$, we define  $A_J^+$ (resp. $A_J^-$) as the set of positive (resp. negative) arcs that are exclusively used by OD-pairs in $J$. For each  OD-pair $\ell \in L$, define $r_\ell^+$ (resp. $r_\ell^-$) to be the unique positive (resp. negative) route connecting the origin of $\ell$ to its destination. Then $a \in A_J^+$ if $a \in r_\ell^+$ for all $\ell \in J$ and $a \notin r_\ell^+$ for all $\ell \in L \setminus J$. We proceed similarly for $A_J^-$.
We define moreover $A_J=A_J^+\cup A_J^-$. In particular, $A_{\emptyset}$ is the set of arcs contained in no route. The sets $A_J$ form a partition of the set $A$ of arcs of $G$.

Defining the positive direction as the counterclockwise one on Figure~\ref{fig:mixed}, we have $$A_{\{(o_1,d_1),(o_2,d_2)\}}^+=\{(o_2,v),(v,d_1)\}$$ $$A_{\{(o_2,d_2)\}}^-=\{(o_2,o_1)\}$$ $$A_{\emptyset}=\{(d_1,v),(v,o_2),(d_2,w),(w,u),(u,o_1)\}.$$ 

The sets $A_J^{\varepsilon}$ enjoy three useful properties.

\begin{claim}\label{claim:exist} For any $\varepsilon\in\{-,+\}$ and any $\ell \in L$, there is at least one $J\subseteq L$ containing $\ell$ such that $A_J^{\varepsilon}$ is nonempty.
\end{claim}
\begin{proof}  Indeed, there is at least one arc of $G$ on the route $r_\ell^\varepsilon$. 
\end{proof} 

\begin{claim}\label{claim:comp}For any $J\subseteq L$, we have $A_J^+\neq\emptyset$ if and only if $A_{L\setminus J}^-\neq\emptyset.$
\end{claim}
\begin{proof} 
It is a consequence of Claim~\ref{claim:sum}: if $a^+\in A_J^+$, then $a^-\in A_{L \setminus J}^-$.
\end{proof} 

\begin{claim}\label{claim:sep}For any distinct $\ell$ and $\ell'$ in $L$, there is at least one $J$ such that $|\{\ell,\ell'\}\cap J|=1$ and $A_J\neq\emptyset$. 
\end{claim}
\begin{proof} 
Indeed, let $\ell=(o,d)$ and $\ell'=(o',d')$ be two distinct OD-pairs of $H$. Since $H$ is simple, it contains no multiple edges and $o\neq o'$ or $d\neq d'$. It means that there is at least one arc of $G$ which is in exactly one of the four $(o,d)$- and $(o',d')$-routes.
\end{proof}

\subsection{If each arc of $G$ is contained in at most two routes, the uniqueness property holds}\label{subsec:uniq}
For each user $i$, we define $r_i^+$ (resp. $r_i^-$) to be the unique positive (resp. negative) route connecting the origin of $i$ to its destination.
 For a strategy profile $\sigma$ and a subset $J \subseteq L$, we define $f_J^+$ and $f_J^-$ to be:
$$f_J^+ = \int_{i\in\bigcup_{\ell\in J}I_{\ell}}\ind{\sigma(i)=r_i^+}d\lambda \ \text{ and } \ f_J^-=\int_{i\in\bigcup_{\ell\in J}I_{\ell}}\ind{\sigma(i)=r_i^-}d\lambda.$$
The quantity $f_J^+$ (resp. $f_J^-$) is thus the number of users $i$ in a $I_{\ell}$ with $\ell\in J$ choosing a positive (resp. negative) route. Note that the quantity $f_J^++f_J^-=\sum_{\ell\in J}\lambda(I_{\ell})$ does not depend on the strategy $\sigma$.

Assume that we have two distinct equilibria $\sigma$ and $\hat{\sigma}$. The flows induced by $\hat\sigma$ are denoted with a hat: $\f$. We define for any subset $J \subseteq L$: 
\begin{equation}\label{eq:fJ}
\Delta_J=f_J^+-\f_J^+=\f_J^--f_J^- .
\end{equation}
By a slight abuse of notation, we let $\Delta_{\ell}:=\Delta_{\{\ell\}}$ for $\ell\in L$.  

For each user $i$, we define $\delta(i) = \ind{\sigma(i)=r_i^+} - \ind{\hat{\sigma}(i)=r_i^+}=\ind{\hat{\sigma}(i)=r_i^-} - \ind{\sigma(i)=r_i^-}$. Then, the following lemma holds.
\begin{lemma}\label{lem:alter}
Let $\ell\in L$ and $i\in I_{\ell}$ be such that $\delta(i)\neq 0$. Then exactly one of the following alternatives holds.
\begin{itemize}
\item There is a $J\subseteq L$ with $\ell\in J$, $A_J\neq\emptyset$, and $\delta(i)\Delta_J<0$.
\item For all $J\subseteq L$ with $\ell\in J$ and $A_J\neq\emptyset$, we have $\Delta_J=0$.
\end{itemize}
\end{lemma}

We briefly explain the intuition behind this lemma. Assume that we move from $\sigma$ to $\hat\sigma$. If a user $i$ changes his chosen route, we are in one of the following two situations. 

The first situation is when the cost of the new route decreases or the cost of the old route increases. If the cost of a route decreases (resp. increases), there is at least one arc of this route whose flow decreases (resp. increases). Since an arc belongs to some set $A_J^\eps$, we get the first point of Lemma~\ref{lem:alter}. 

The second situation is when the costs remain the same for both routes and both routes have same costs, which implies the second point of Lemma~\ref{lem:alter}. 

\begin{proof}[Proof of Lemma~\ref{lem:alter}]
As $\sigma$ is an equilibrium, we have for each user $i$:
\begin{equation}\label{eq:equi}\sum_{a\in A} c_a^i(f_a) \left( \ind{a \in \sigma(i)} - \ind{a \in \hat{\sigma}(i)} \right) \leq 0.\end{equation}
For $a \in A_J^+$, we have $f_a=f_J^+$ and $\ind{a \in \sigma(i)}=\ind{\sigma(i)=r_i^+}\ind{a \in r_i^+}$, and the same holds for $a \in A_J^-$. By decomposing the sum~\eqref{eq:equi}, we obtain that
$$\sum_{J\subseteq L}\left(\sum_{a \in A^+_J \cap r_i^+}c_a^i(f_J^+)\delta(i)-\sum_{a \in A^-_J \cap r_i^-}c_a^i(f_J^-)\delta(i)\right)\leq 0.$$
We can write a similar equation for the equilibrium $\hat{\sigma}$. By summing them, we obtain

\begin{equation}\label{eq:twoeq}
\delta(i) \sum_{J\subseteq L}  \left(\sum_{a \in A^+_J\cap r_i^+} (c_a^i(f_J^+)-c_a^i(\f_J^+))- \sum_{a \in A^-_J \cap r_i^-}(c_a^i(f_J^-)-c_a^i(\f_J^-))\right)\leq 0.
\end{equation}

According to Equation~\eqref{eq:fJ} and using the fact that the maps $c_a^i$ are strictly increasing, both
$\sum_{a \in A^+_J \cap r_i^+}(c_a^i(f_J^+)-c_a^i(\f_J^+))$ and $-\sum_{a \in A^-_J \cap r_i^-}(c_a^i(f_J^-)-c_a^i(\f_J^-))$ have the sign of $\Delta_J$. Therefore, if all terms of the sum in Equation~\eqref{eq:twoeq} are equal to $0$, the second point of the lemma holds. If at least one term of the sum is $<0$, we get the first point.
\end{proof} 

With the help of this lemma, we get one direction of Theorem~\ref{thm:main}.

\begin{proposition}\label{prop:uniq}
If each arc of $G$ is contained in at most two routes, the uniqueness property holds.
\end{proposition}
\begin{proof}
Note that the assumption of the proposition ensures that $A_J=\emptyset$ if $|J|\geq 3$. We want to prove that $\Delta_J=0$ for all $J\subseteq L$ such that $A_J\neq\emptyset$.

Assume for a contradiction that there is a $J_0$ such that $\Delta_{J_0}\neq 0$ and $A_{J_0}\neq\emptyset$. Then there is a $\ell_0 \in {J_0}$ such that $\Delta_{\ell_0} \neq 0$. At least one user $i_0\in I_{\ell_0}$ is such that $\delta(i_0)\Delta_{\ell_0}>0$. 

Suppose that the first case of Lemma~\ref{lem:alter} occurs. There exists $\ell_1 \in L$, $\ell_1 \neq \ell_0$ with $A_{\{\ell_0,\ell_1\}} \neq \emptyset$ and $\delta(i_0)\Delta_{\{\ell_0,\ell_1\}}<0$. Then, $\delta(i_0)\Delta_{\{\ell_0,\ell_1\}}=\delta(i_0)(\Delta_{\ell_0}+\Delta_{\ell_1})<0$, which implies that $|\Delta_{\ell_0}|<|\Delta_{\ell_1}|$. 
It follows that $\Delta_{\ell_1}\neq 0$, and taking $i_1 \in I_{\ell_1}$ with $\delta(i_1)\Delta_{\ell_1} >0$, only the first case of Lemma~\ref{lem:alter} can occur for $i=i_1$ and $\ell=\ell_1$. Indeed, the second case would imply that $\Delta_{\{\ell_0,\ell_1\}}=0$ since $A_{\{\ell_0,\ell_1\}} \neq \emptyset$.
Repeating the same argument, we build an infinite sequence $(\ell_0,\ell_1, \ldots)$ of elements of $L$ such that, for each $k\geq 0$, $A_{\{\ell_k,\ell_{k+1}\}}\neq \emptyset$ and $|\Delta_{\ell_k}|<|\Delta_{\ell_{k+1}}|$. This last condition implies that the $\ell_k$ are distinct, which is impossible since $|L|$ is finite.

Thus, the second case of Lemma~\ref{lem:alter} occurs for $\ell_0$, and 
hence $\Delta_{J_0}=0$, which is in contradiction with the starting assumption. On any arc, we have a total flow that remains the same when changing from $\sigma$ to $\hat\sigma$.
\end{proof}

 The only fact we use from the ring structure is that there are two sets $A^+$ (positive arcs) and $A^-$ (negative arcs) and that each user has exactly two possible strategies, each of them being included in one of these two sets. We can state a result holding for more general nonatomic congestion game with user-specific cost functions. We omit the proof since the one of Proposition~\ref{prop:uniq} holds without any change.

\begin{proposition}\label{prop:more2}
Consider a nonatomic congestion game with user-specific (strictly increasing) cost functions. Let $A^+$ and $A^-$ be two disjoint finite sets. Assume that every user $i$ has exactly two available strategies $r_i^+$ and $r_i^-$ with $r_i^+\subseteq A^+$ and $r_i^-\subseteq A^-$. Then, if all triples of pairwise distinct strategies have an empty intersection, the uniqueness property holds. 
\end{proposition}

\subsection{If an arc of $G$ is contained in at least three routes, a counterexample exists}\label{subsec:nonuniq}

We give an explicit construction of multiple equilibrium flows when an arc is contained in at least three routes.

\subsubsection{If $|L|=3$}\label{sec:construction} In order to ease the notation, we use $1$, $2$, and $3$ to denote the three OD-pairs of $H$. We denote accordingly by $I_1$, $I_2$, and $I_3$ the three sets of users associated to each of these OD-pairs.

We can assume without loss of generality that $A_{\{1,2,3\}}^+\neq\emptyset$, $A_{\{1,2\}}\neq\emptyset$, and $A_{\{1,3\}}\neq\emptyset$.
The first assumption can be done since there is an arc in three routes. For the other ones: with the help of Claim~\ref{claim:sep}, and if necessary of Claim~\ref{claim:comp}, we get that there is at least a $J$ of cardinality two such that $A_J\neq\emptyset$. Again, using Claim~\ref{claim:sep}, this time with the two elements of $J$, and if necessary Claim~\ref{claim:comp}, we get another $J'$ of cardinality two such that $A_{J'}\neq\emptyset$.

\paragraph{Definition of the cost functions}
We define three classes of users. Each of these classes is attached to one of the OD-pairs. For a class $k\in\{1,2,3\}$, we define the cost functions $c_J^{k,\eps}$, for all $J\subseteq\{1,2,3\}$ and $\eps\in\{-,+\}$. The cost function for a class $k$ user $i$ on an arc $a$ of $A_J^{\eps}$ is set to $c_a^i:=c_J^{k,\eps}$. If the set $A_J^{\eps}$ is empty, the definition of $c_J^{k,\eps}$ is simply discarded.

\begin{description}
\item[\textbf{Class $1$}] We define this class to be the users of the set $I_1$. We set $\lambda(I_1)=1.5$ and choose $J_1\subseteq\{1,2,3\}$ with $1\in J_1$ such that $A_{J_1}^-\neq\emptyset$ (with the help of Claim~\ref{claim:exist}).

$$\left\{
\begin{array}{l}
\displaystyle{c_{\{1,2,3\}}^{1,+}(x)=\frac{24x+7}{|A_{\{1,2,3\}}^+|}} \\
\displaystyle{c_J^{1,+}(x)=\frac{x}{|A_J^+|}\quad\mbox{for any $J\neq\{1,2,3\}$ with $1\in J$}} \\
\displaystyle{c_{J_1}^{1,-}(x)=\frac{x+48}{|A_{J_1}^-|}} \\
\displaystyle{c_J^{1,-}(x)=\frac{x}{|A_J^-|}\quad\mbox{for any $J\neq J_1$ with $1\in J$.}}
\end{array}\right.$$

\bigskip

\item[\textbf{Class $2$}] We define this class to be the users of the set $I_2$. We set $\lambda(I_2)=1$. We have assumed that $A_{\{1,2\}}\neq\emptyset$. We distinguish hereafter the cases $A_{\{1,2\}}^+\neq\emptyset$ and $A_{\{1,2\}}^-\neq\emptyset$ (which may hold simultaneously, in which case we make an arbitrary choice).

\begin{description} 
\item[\textbf{If $A_{\{1,2\}}^+\neq\emptyset$}] We choose $J_2\subseteq\{1,2,3\}$ with $2\in J_2$ such that $A_{J_2}^-\neq\emptyset$ (with the help of Claim~\ref{claim:exist}).
$$\left\{
\begin{array}{l}
\displaystyle{c_{\{1,2\}}^{2,+}(x)=\frac{25x}{|A_{\{1,2\}}^+|}} \\
\displaystyle{c_J^{2,+}(x)=\frac{x}{|A_J^+|}\quad\mbox{for any $J\neq\{1,2\}$ with $2\in J$}} \\
\displaystyle{c_{J_2}^{2,-}(x)=\frac{x+31}{|A_{J_2}^-|}} \\
\displaystyle{c_J^{2,-}(x)=\frac{x}{|A_J^-|}\quad\mbox{for any $J\neq J_2$ with $2\in J$.}}
\end{array}\right.$$

\item[\textbf{If $A_{\{1,2\}}^-\neq\emptyset$}]
 $$\left\{
 \begin{array}{l}
 \displaystyle{c_{\{1,2,3\}}^{2,+}(x)=\frac{x+26}{|A_{\{1,2,3\}}^+|}} \\
 \displaystyle{c_J^{2,+}(x)=\frac{x}{|A_J^+|}\quad\mbox{for any $J\neq\{1,2,3\}$ with $2\in J$}} \\
 \displaystyle{c_{\{1,2\}}^{2,-}(x)=\frac{22x}{|A_{\{1,2\}}^-|}} \\
 \displaystyle{c_{J}^{2,-}(x)=\frac{x}{|A_J^-|}\quad\mbox{for any $J\neq\{1,2\}$ with $2\in J$.}}
 \end{array}\right.$$
\end{description}

\bigskip

\item[\textbf{Class $3$}] We define this class to be the users of the set $I_3$. We set $\lambda(I_3)=1$. We have assumed that $A_{\{1,3\}}\neq\emptyset$. We distinguish hereafter the cases $A_{\{1,3\}}^+\neq\emptyset$ and $A_{\{1,3\}}^-\neq\emptyset$ (which may hold simultaneously, in which case we make an arbitrary choice).

\begin{description} 
\item[\textbf{If $A_{\{1,3\}}^+\neq\emptyset$}] We choose $J_3\subseteq\{1,2,3\}$ with $3\in J_3$ such that $A_{J_3}^-\neq\emptyset$ (with the help of Claim~\ref{claim:exist}).
$$\left\{
\begin{array}{l}
\displaystyle{c_{\{1,3\}}^{3,+}(x)=\frac{25x}{|A_{\{1,3\}}^+|}} \\
\displaystyle{c_J^{3,+}(x)=\frac{x}{|A_J^+|}\quad\mbox{for any $J\neq\{1,3\}$ with $3\in J$}} \\
\displaystyle{c_{J_3}^{3,-}(x)=\frac{x+31}{|A_{J_3}^-|}} \\
\displaystyle{c_{J}^{3,-}(x)=\frac{x}{|A_J^-|}\quad\mbox{for any $J\neq J_3$ with $3\in J$.}}
\end{array}\right.$$

\item[\textbf{If $A_{\{1,3\}}^-\neq\emptyset$}]
$$\left\{
\begin{array}{l}
\displaystyle{c_{\{1,2,3\}}^{3,+}(x)=\frac{x+26}{|A_{\{1,2,3\}}^+|}} \\
\displaystyle{c_J^{3,+}(x)=\frac{x}{|A_J^+|}\quad\mbox{for any $J\neq\{1,2,3\}$ with $3\in J$}} \\
\displaystyle{c_{\{1,3\}}^{3,-}(x)=\frac{22x}{|A_{\{1,3\}}^-|}} \\
\displaystyle{c_{J}^{3,-}(x)=\frac{x}{|A_J^-|}\quad\mbox{for any $J\neq\{1,3\}$ with $3\in J$.}}
\end{array}\right.$$
\end{description}

\end{description}

\bigskip

\paragraph{Definition of two strategy profiles}
We define now two strategy profiles $\sigma$ and $\hat\sigma$, inducing distinct flows on some arcs. We check in the next paragraph that each of them is an equilibrium.

\medskip

\begin{description}
\item[Strategy profile $\sigma$] For all $i\in I_1$, we set $\sigma(i)=r_i^+$ and for all $i\in I_2\cup I_3$, we set $\sigma(i)=r_i^-$. Then, the flows are the following:
$$
\begin{array}{c|ccccccc}
\hline\noalign{\smallskip}
   J & \{1\} & \{2\} & \{3\} & \{1,2\} & \{1,3\} & \{2,3\} & \{1,2,3\} \\ 
      \noalign{\smallskip}\hline\noalign{\smallskip}
f_J^+ & 1.5 & 0& 0& 1.5 & 1.5 & 0 & 1.5  \\
f_J^- & 0 & 1& 1& 1 & 1 & 2 & 2\\
\noalign{\smallskip}\hline
\end{array}
$$

\item[Strategy profile $\hat\sigma$] For all $i\in I_1$, we set $\hat\sigma(i)=r_i^-$ and for all $i\in I_2\cup I_3$, we set $\hat\sigma(i)=r_i^+$. Then, the flows are the following:
$$
\begin{array}{c|ccccccc}
\hline\noalign{\smallskip}
   J & \{1\} & \{2\} & \{3\} & \{1,2\} & \{1,3\} & \{2,3\} & \{1,2,3\} \\ 
   \noalign{\smallskip}\hline\noalign{\smallskip}
\f_J^+ & 0 & 1& 1& 1 & 1 & 2 & 2 \\
\f_J^- & 1.5 & 0& 0& 1.5 & 1.5 & 0 & 1.5 \\
\noalign{\smallskip}\hline
\end{array}
$$
\end{description}

\paragraph{The strategy profiles are equilibria}
We check now that $\sigma$ and $\hat{\sigma}$ are equilibria, by computing the cost of each of the two possible routes for each class.

For a class $k \in \{1,2,3\}$, we denote with a slight abuse of notation the common positive (resp. negative) route of the class $k$ users by $r_k^+$ (resp. $r_k^-$).
 \begin{description}
  \item[\textbf{Class $1$}] 
We put in the following tables, the costs experienced by the class $1$ users on the various arcs of $G$ for each of $\sigma$ and $\hat\sigma$. For a given $J\subseteq\{1,2,3\}$ with $1\in J$ and $\varepsilon\in\{-,+\}$, we indicate the cost experienced by any class $1$ user on the whole collection of arcs in $A_J^{\varepsilon}$. For instance in $\sigma$, if $J=\{1,2,3\}$, then $f_J^+ = 1.5$, and the cost of all arcs together in $A_{J}^+$ is $|A_{J}^+| c_J^{1,+}(1.5)=43$. 

\medskip

 For the strategy profile $\sigma$, we get the following flows and costs on the arcs of $G$ for a class $1$ user.\\

\begin{center}
\begin{tabular}{c|cc|cc}
\multicolumn{1}{c}{} & \multicolumn{2}{c}{$\eps=+$} & \multicolumn{2}{c}{$\eps=-$} \\
\hline\noalign{\smallskip}
$J$ with $1\in J$	&  \{1,2,3\} 	& other & $J_1$ & other\\ 
\noalign{\smallskip}\hline\noalign{\smallskip}
$f_J^{\eps}$		&  1.5 		&  1.5	&  0, 1, or 2	& 0, 1, or 2	 \\
Cost on $A_J^{\eps}$	& 43 		& 1.5	& 48, 49, or 50 & 0, 1, or 2 \\
\noalign{\smallskip}\hline
\end{tabular}
\end{center}
\medskip
Using the fact that $A_{\{1,2,3\}}^+\neq\emptyset$, the total cost of $r_1^+$ in $\sigma$ for a class $1$ user is equal to 
$$43+1.5\times \left|\{\mbox{$J\neq\{1,2,3\}$ such that $A_J^+\neq\emptyset$ and $1\in J$}\}\right|.$$ Since there are at most three sets $J\neq\{1,2,3\}$ such that $A_J^+\neq\emptyset$ and $1\in J$, we get that the total cost of $r_1^+$ lies in $[43;47.5]$. Similarly, using the fact that $A_{J_1}^-\neq\emptyset$, we get that the total cost of $r_1^-$ for a class $1$ user lies in $[48;54]$. Therefore the users of class $1$ are not incitated to change their choice in $\sigma$.

\medskip

For the strategy profile $\hat\sigma$, we get the following flows and costs.\\

\begin{center}
\begin{tabular}{c|cc|cc}
\multicolumn{1}{c}{} & \multicolumn{2}{c}{$\eps=+$} & \multicolumn{2}{c}{$\eps=-$} \\
\hline\noalign{\smallskip}
$J$ with $1\in J$	&  \{1,2,3\} 	& other & $J_1$ & other\\ 
\noalign{\smallskip}\hline\noalign{\smallskip}
$\hat{f}_J^{\eps}$		&  2 		&  0 or 1	&  1.5		& 1.5		 \\
Cost on $A_J^{\eps}$	& 55 		& 0 or 1	& 49.5		& 1.5 \\
\noalign{\smallskip}\hline
\end{tabular}
\end{center}
\medskip
The total cost of $r_1^+$ for a class $1$ user lies in $[55;58]$ and the total cost of $r_1^-$ for a class $1$ user lies in $[49.5;54]$. Therefore the users of class $1$ are not incitated to change their choice in $\hat\sigma$. 

\bigskip

   \item[\textbf{Class $2$}] 
   \begin{description}
\item[\textbf{If $A_{\{1,2\}}^+\neq\emptyset$}] We put in the following tables, the costs experienced by the class $2$ users on the various arcs of $G$ for each of $\sigma$ and $\hat\sigma$. 

\medskip

For the strategy profile $\sigma$: \\
 \begin{center}
\begin{tabular}{c|cc|cc}
\multicolumn{1}{c}{} & \multicolumn{2}{c}{$\eps=+$} & \multicolumn{2}{c}{$\eps=-$} \\
\hline\noalign{\smallskip}
$J$ with $2\in J$	&  \{1,2\} 	& other & $J_2$ & other\\ 
\noalign{\smallskip}\hline\noalign{\smallskip}
$f_J^{\eps}$		&  1.5 		&  0 or 1.5	&  1 or 2	& 1 or 2	 \\
Cost on $A_J^{\eps}$	& 37.5 		& 1.5	& 32 or 33 & 1 or 2 \\
\noalign{\smallskip}\hline
\end{tabular}
\end{center}
\medskip
The total cost of $r_2^+$ for a class $2$ user is precisely $39$ (we use the fact that $A_{\{1,2,3\}}^+\neq\emptyset$) and the total cost of $r_2^-$ lies in $[32;38]$. The users of class $2$ are not incitated to change their choice in $\sigma$.

\medskip

For the strategy profile $\hat\sigma$:\\

\begin{center}
\begin{tabular}{c|cc|cc}
\multicolumn{1}{c}{} & \multicolumn{2}{c}{$\eps=+$} & \multicolumn{2}{c}{$\eps=-$} \\
\hline\noalign{\smallskip}
$J$ with $2\in J$	&  \{1,2\} 	& other & $J_2$ & other\\
\noalign{\smallskip}\hline\noalign{\smallskip}
$\hat{f}_J^{\eps}$		&  1 		&  1 or 2	&  0 or 1.5	& 0 or 1.5	 \\
Cost on $A_J^{\eps}$	& 25 		& 1 or 2	& 31 or 32.5 & 0 or 1.5\\
\noalign{\smallskip}\hline
\end{tabular}
\end{center}
\medskip
The total cost of $r_2^+$ for a class $2$ user lies in $[27;30]$ and the total cost of $r_2^-$ lies in $[31;34]$. The users of class $2$ are not incitated to change their choice in $\hat\sigma$. \\

\item[\textbf{If $A_{\{1,2\}}^-\neq\emptyset$}] We put in the following tables, the costs experienced by the class $2$ users on the various arcs of $G$ for each of $\sigma$ and $\hat\sigma$. 

\medskip

For the strategy profile $\sigma$: \\

\begin{center}
\begin{tabular}{c|cc|cc}
\multicolumn{1}{c}{} & \multicolumn{2}{c}{$\eps=+$} & \multicolumn{2}{c}{$\eps=-$} \\
\hline\noalign{\smallskip}
$J$ with $2\in J$	&  \{1,2,3\} 	& other & $\{1,2\}$ & other\\ 
\noalign{\smallskip}\hline\noalign{\smallskip}
$f_J^{\eps}$		&  1.5 		&  0 or 1.5	&  1	& 1 or 2	 \\
Cost on $A_J^{\eps}$	& 27.5 		& 0 or 1.5	& 22 & 1 or 2 \\
\noalign{\smallskip}\hline
\end{tabular}
\end{center}
\medskip
The total cost of $r_2^+$ for a class $2$ user lies in $[27.5;29]$ and the total cost of $r_2^-$ lies in $[22;27]$. The users of class $2$ are not incitated to change their choice in $\sigma$.

\medskip

For the strategy profile $\hat\sigma$:\\ 

\begin{center}
\begin{tabular}{c|cc|cc}
\multicolumn{1}{c}{}& \multicolumn{2}{c}{$\eps=+$} & \multicolumn{2}{c}{$\eps=-$} \\
\hline\noalign{\smallskip}
$J$ with $2\in J$	&  \{1,2,3\} 	& other & $\{1,2\}$ & other\\ 
\noalign{\smallskip}\hline\noalign{\smallskip}
$\hat{f}_J^{\eps}$		&  2 		&  1 or 2	&  1.5	& 0 or 1.5	 \\
Cost on $A_J^{\eps}$	& 28 		& 1 or 2	& 33 & 0 or 1.5 \\
\noalign{\smallskip}\hline
\end{tabular}
\end{center}
\medskip
The total cost of $r_2^+$ for a class $2$ user lies in $[28;32]$ and the total cost of $r_2^-$ lies in $[33;34.5]$. The users of class $2$ are not incitated to change their choice in $\hat\sigma$. 
 \end{description}
 
 \bigskip
 
    \item[\textbf{Class $3$}] 
The symmetry of the cost functions for classes $2$ and $3$ gives the same tables for class $3$ as for class $2$, by substituting $\{1,3\}$ to $\{1,2\}$. Therefore, we get the same conclusions: neither in $\sigma$, nor in $\hat\sigma$, the class $3$ users are incitated to change their choice.
  \end{description}

\bigskip

Therefore, $\sigma$ and $\hat\sigma$ are equilibria and induce distinct flows. It proves that the uniqueness property does not hold. It remains to check the case when $|L|>3$.

\begin{remark}
 A classical question when there are several equilibria is whether one of them dominates the others. An equilibrium is said to {\em dominate} another one if it is preferable for all users. In this construction, no equilibrium dominates the other, except when $A_{\{1,2,3\}}^+\neq\emptyset$, $A_{\{1,2\}}^-\neq\emptyset$, and $A_{\{1,3\}}^-\neq\emptyset$ where $\sigma$ dominates $\hat \sigma$.
\end{remark}

\subsubsection{If $|L|> 3$} Denote $1$, $2$, and $3$ three OD-pairs of $H=(T,L)$ giving three routes containing the same arc of $G$. For these three arcs of $H$, we make the same construction as above, in the case $|L|=3$. For the other $\ell\in L$, we set $I_{\ell}=\emptyset$ to get the desired conclusion.

However, note that we can also get multiple equilibrium flows, while requiring $I_{\ell}\neq\emptyset$ for all $\ell\in L$. For $\ell\notin\{1,2,3\}$, we use a fourth class, whose costs are very small on all positive arcs of $G$ and very large on all negative arcs of $G$, and whose measure is a small positive quantity $\delta$. Each user of this class chooses always a positive route, whatever the other users do. For $\delta$ small enough, the users of this class have no impact on the choices of the users of the classes $1$, $2$, and $3$, as the difference of cost between the routes is always bounded below by $0.5$.

\section{When the supply graph is a ring having each arc in at most two routes}\label{sec:comb}

In this section, we provide a further combinatorial analysis of the characterization of the uniqueness property for ring graphs stated in Theorem~\ref{thm:main}.

\subsection{Two corollaries}\label{subsec:cor} 

\begin{corollary}\label{cor:ok2}
If the supply graph is a cycle and if there are at most two OD-pairs, i.e. $|L|\leq 2$, then the uniqueness property holds.
\end{corollary}
 
\begin{corollary}\label{cor:no5}
If the supply graph is a cycle and if the uniqueness property holds, then the number of OD-pairs is at most $4$, i.e. $|L|\leq 4$.
\end{corollary}

Corollary~\ref{cor:ok2} is straightforward. Corollary~\ref{cor:no5} is a direct consequence of Claim~\ref{claim:sum}: if $|L|\geq 5$, then there is necessarily an arc of $G$ in three routes.

\subsection{How to compute in one round trip the maximal number of routes containing an arc of $G$}\label{subsec:algo} In an arc $(u,v)$, vertex $u$ is called the {\em tail} and vertex $v$ is called the {\em head}. 
The algorithm starts at an arbitrary vertex of $H$ and makes a round trip in an arbitrary direction, while maintaining a triple \texttt{(list, min, max)}. In this triple, \texttt{list} is a set of arcs of $H$ whose tail has already been encountered but whose head has not yet been encountered. At the beginning, \texttt{list} is empty and \texttt{min}  and \texttt{max}  are both zero. When the algorithm encounters a vertex $v$, it proceeds to three operations. \\

{\bf First operation.} It computes the number of arcs in $\delta_H^-(v)$ (arcs of $H$ having $v$ as head) not in \texttt{list}. The corresponding routes were ``forgotten'', since the tail is ``before'' the starting vertex of the algorithm. To take them into account, this number is added to the values \texttt{min} and \texttt{max}. 

{\bf Second operation.} All arcs of \texttt{list} being also in $\delta_H^-(v)$ are removed from \texttt{list}.

{\bf Third operation.} All arcs in $\delta_H^+(v)$ (arcs of $H$ having $v$ as tail) are added to \texttt{list}. The value of \texttt{min} is updated to the minimum between the previous value of \texttt{min} and the size of \texttt{list}, and similarly \texttt{max} is updated to the maximum between the previous value of \texttt{max} and the size of \texttt{list}. \\

The algorithm stops after one round trip. At the end, the values \texttt{min} and \texttt{max} are respectively the minimal and maximal number of routes containing an arc in the direction chosen. According to Claim~\ref{claim:sum}, $\max(|L|-\texttt{min},\texttt{max})$ is the maximal number of routes containing an arc of $G$.

Note that with a second round trip, this algorithm can specify the routes containing a given arc, by scanning the content of \texttt{list}.

\subsection{Explicit description of the networks having the uniqueness property when the supply graph is a cycle}\label{subsec:explicit}

\begin{proposition}\label{prop:carac} Let the supply graph $G$ be a cycle. Then, for any demand digraph $H$, the pair $(G,H)$ is such that each arc of $G$ is in at most two routes if and only if the mixed graph $G+H$ is homeomorphic to a minor of one of the nine mixed graphs of Figures~\ref{fig:min1}--\ref{fig:min4}.
\end{proposition}

Combined with Theorem~\ref{thm:main}, this proposition allows to describe explicitely all pairs $(G,H)$ having the uniqueness property, when $G$ is a cycle.\\

\begin{proof} One direction is straightforward. Let us prove the other direction, namely that, if each arc of $G$ is in at most two routes, then $G+H$ is homeomorphic to a minor of one of the nine mixed graphs. We can assume that $V=T$. According to Corollary~\ref{cor:no5}, we can also assume that $|L|\in\{1,2,3,4\}$. \\

If $|L|\in\{1,2\}$, there is nothing to prove: all possible mixed graphs with $|L|=1$ or $|L|=2$ are homeomorphic to a minor of the graphs of Figures~\ref{fig:min1} and~\ref{fig:min2}.\\

If $|L|=3$, we can first assume that $L$ contains two disjoint arcs that are crossing in the plane embedding. By trying all possibilities for the third arc, we get that the only possible configuration is the right one on Figure~\ref{fig:min3} and the ones obtained from it by edge contraction. Second, we assume that there are no ``crossing'' arcs. The three heads of the arcs cannot be consecutive on the cycle otherwise we would have an arc of $G$ in three routes. Again by enumerating all possibilities, we get that the only possible configuration is the left one on Figure~\ref{fig:min3} and the ones obtained from it by edge contraction.\\

If $|L|=4$, Claim~\ref{claim:comp} shows that each arc of $G$ belongs to exactly two routes. It implies that, in $H$, the indegree of any $\ell\in L$ is equal to its outdegree. There are therefore circuits in $H$. It is straightforward to check that it is impossible to have a length $3$ circuit. It remains to enumerate the possible cases for length $2$ and length $4$ circuits to get that the only possible configurations are the ones of Figure~\ref{fig:min4} and the common one obtained from them by edge contraction.
\end{proof} 

 \begin{figure}\begin{center}
\includegraphics[scale=0.9]{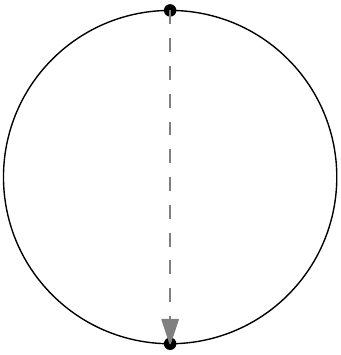}
\caption{All rings with $|L|=1$ (i.e. one OD-pair) having the uniqueness property are homeomorphic to this graph}
\label{fig:min1}
 \end{center}\end{figure}
 
\begin{figure}\begin{center} 
\includegraphics[scale=0.9]{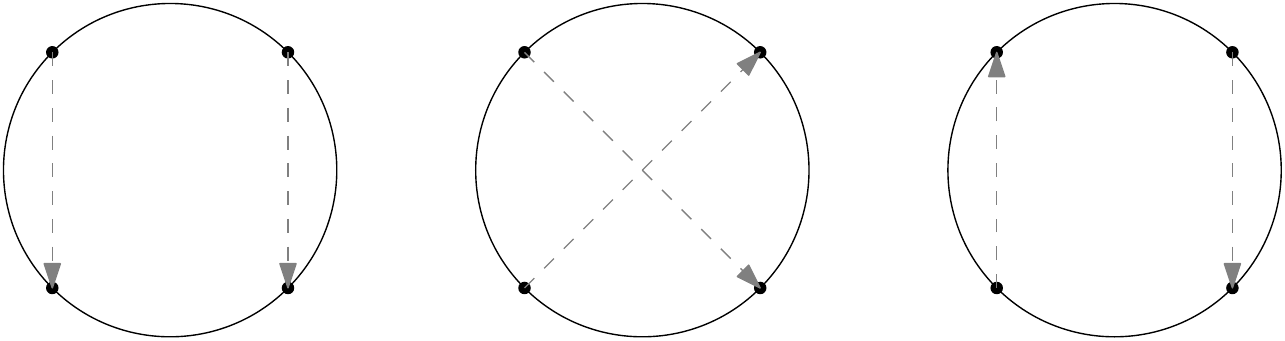}
\caption{All rings with $|L|=2$ (i.e. two OD-pairs) having the uniqueness property are homeomorphic to one or to minors of these graphs}
\label{fig:min2}
\end{center}\end{figure} 

 \begin{figure}\begin{center}
\includegraphics[scale=0.9]{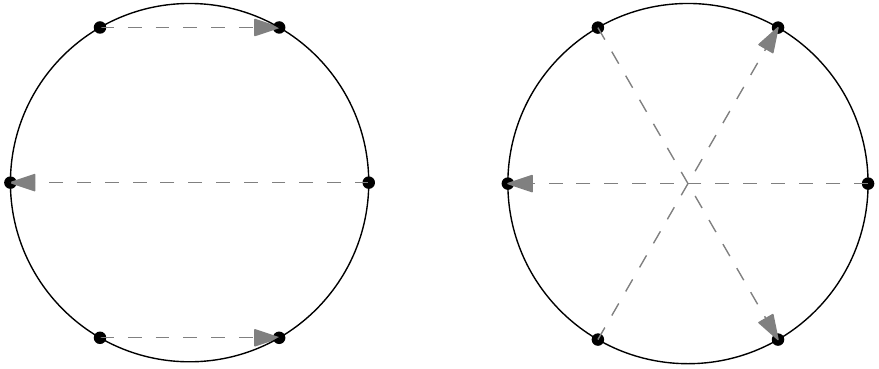}
\caption{All rings with $|L|=3$ (i.e. three OD-pairs) having the uniqueness property are homeomorphic to one or to minors of these graphs}
\label{fig:min3}
\end{center}\end{figure} 

\begin{figure}\begin{center} 
\includegraphics{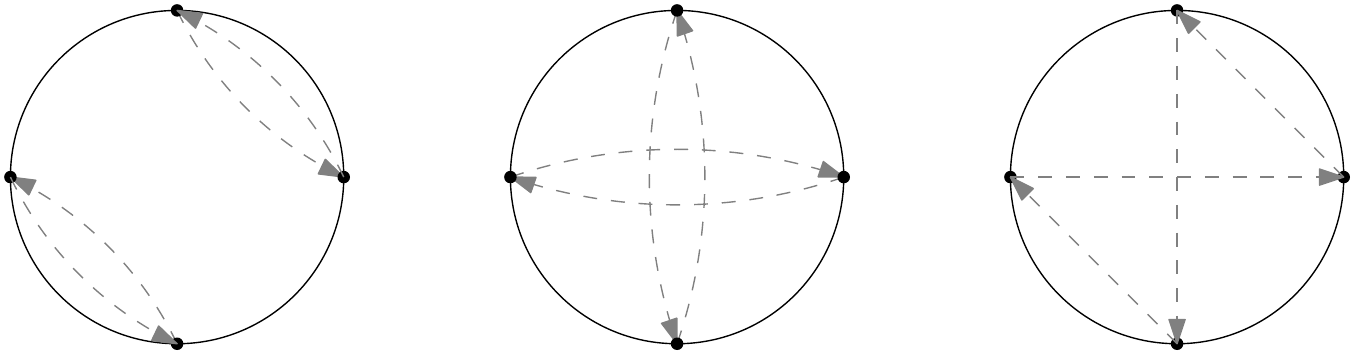}
\caption{All rings with $|L|=4$ (i.e. four OD-pairs) having the uniqueness property are homeomorphic to one or to minors of these graphs}
\label{fig:min4}
\end{center}\end{figure} 

We can also describe the rings having the uniqueness property by minor exclusion, similarly as in \citet{Mi05}. 

\begin{proposition}\label{prop:exclusion}
Let the supply graph $G$ be a cycle. Then, for any demand digraph $H$, the pair $(G,H)$ is such that there exists an arc of $G$ belonging to at least three routes if and only if $G+H$ has one of the nine mixed graphs of Figure~\ref{fig:annexe} as a minor. 
\end{proposition}

 \begin{proof}[Proof (sketched)]
 Suppose that one of the nine mixed graphs of Figure~\ref{fig:annexe} is a minor of the mixed graph $G+H$. The construction of Section~\ref{sec:construction} shows that we can build two distinct equilibria for this minor where all users of a given class have the same strategy. Then, we can extend this counterexample to the graph $G+H$, see Corollary~\ref{cor:minor} in Section~\ref{subsec:general} holding for more general graphs.
 
  To prove that if an arc of $G$ belongs to at least three routes, then one of the nine mixed graphs of Figure~\ref{fig:annexe} is a minor of $G+H$, we proceed to an explicit, but tedious,  enumeration. We enumerate all possible mixed graphs with $|V|=6$ and $|L|=3$ such that each vertex is the tail or the head of exactly one arc in $L$.	Then, we try all possible sequences of edge contractions leading to mixed graphs satisfying two properties: the demand graph is simple and an arc is in three routes. We keep the mixed graphs such that any additional edge contraction 
leads to a violation of these properties. The details are omitted.
\end{proof}

In particular, the construction in Section~\ref{sec:construction} cannot be simplified by exhibiting a counterexample for each mixed graph of Figure~\ref{fig:annexe}, since the proof of Proposition~\ref{prop:exclusion} needs this tedious enumeration.

\section{Discussion}\label{sec:disc}

\subsection{Results for general graphs}\label{subsec:general}

For the sake of simplicity, given a supply graph $G$ and a demand digraph $H$, we say that the mixed graph $G+H$ has the uniqueness property if the pair $(G,H)$ has it.

Using the results of \citet{Mi05} and Theorem~\ref{thm:main}, we can derive results for more general graphs. Milchtaich suggests to add a fictitious origin, linked to all origins, and similarly for the destinations. If the new graph has the uniqueness property, the original one has it as well. However, this approach cannot be used to prove that a graph does not have the uniqueness property. 
For instance this method allows us to prove that the graph on the left in Figure~\ref{fig:milch} has the uniqueness property, but fails to settle the status of the graph on the right. Indeed, the new graph does not have the uniqueness property, using the result of \citet{Mi05}, but the original one has it, using Theorem~\ref{thm:main}.

A way for proving that a pair $(G,H)$ does not have the uniqueness property consists in using subgraphs or minors as obstructions to uniqueness property. If $G+H$ has a subgraph without the uniqueness property, then it does not have the property either. However, it is not clear whether having a minor without uniqueness property is an obstruction for having the uniqueness property. Indeed, the cost functions are strictly increasing and we do not see how in general a counterexample to uniqueness at the level of a minor can be extended at the level of the network itself. Yet, we can settle two specific cases.

The first case is when the contractions involve only bridges of $G$ (a {\em bridge} is an edge whose deletion disconnects the graph). In this case, if the minor does not have the uniqueness property, the pair $(G,H)$ does not have it either. Checking this property is easy.

A second case is formalized in the following proposition. An equilibrium is {\em strict} if each user as a unique best reply.

\begin{proposition}
Let $G'$ and $H'$ be respectively a supply and a demand graphs such that $G'+H'$ is a minor of $G+H$. If there are counterexamples of uniqueness property for $(G',H')$ involving strict equilibria, then $(G,H)$ does not have the uniqueness property.
\end{proposition}

\begin{proof}
We start with a counterexample for $(G',H')$. We de-contract an edge. We assign to this edge a small enough cost function so that the route followed by any user remains a strict best reply for him, whether the route contains the edge or not. Therefore, we can de-contract all edges and get conterexamples to uniqueness property for subgraphs of $G+H$. As noted above, it allows to conclude that $(G,H)$ does not have the uniqueness property.
\end{proof} 

 Since the construction of Section~\ref{sec:construction} provides strict equilibria, we get the following corollary.

\begin{corollary}\label{cor:minor}
 Any mixed graph containing one of the graphs of Figure~\ref{fig:annexe} as a minor does not have the uniqueness property.
\end{corollary}

Let us now give an example using some of these conditions. According to Corollary~\ref{cor:minor} or to the ``bridge-contraction'' condition above, the mixed graph of the Figure~\ref{fig:subgraph} does not have the uniqueness property. Indeed, a ring with an arc in three routes is a minor of it. 
We can conclude that any network having the mixed graph of the Figure~\ref{fig:subgraph} as a subgraph (or as a minor) does not have the uniqueness property. 

However, there are still graphs for which none of the considerations above allows to conclude, see for example the graph of Figure~\ref{fig:unknown}.

 \begin{figure}\begin{center}
\includegraphics[scale=0.9]{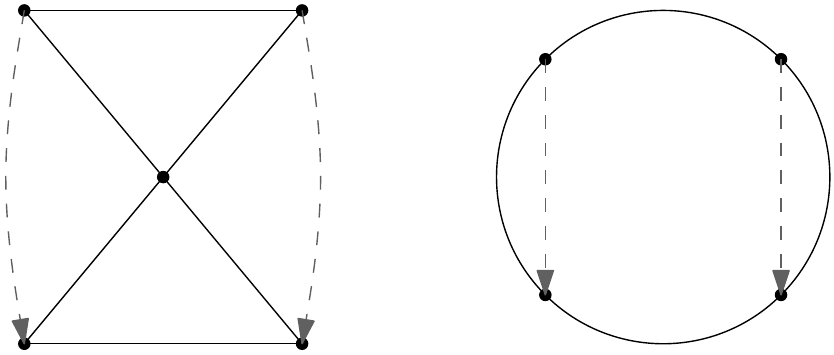}
\caption{Adding a fictitious origin and a fictitious destination settles the case of the left graph but not the case of the right graph}
\label{fig:milch}
 \end{center}\end{figure}

 \begin{figure}\begin{center}
\includegraphics[scale=0.9]{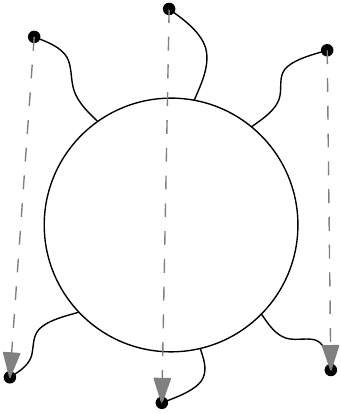}
\caption{Any graph having this one as a subgraph does not have the uniqueness property}
\label{fig:subgraph}
 \end{center}\end{figure}

\begin{figure}\begin{center} 
\includegraphics[scale=1]{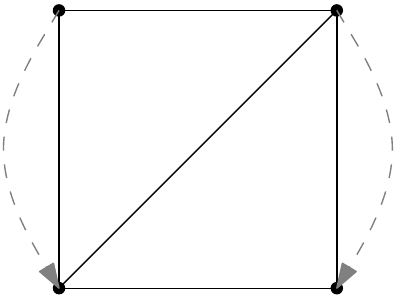}
\caption{A graph for which neither Theorem~\ref{thm:main} nor \citet{Mi05} can be used to prove or disprove the uniqueness property}
\label{fig:unknown}
 \end{center}\end{figure}

%


\subsection{Equivalence of equilibria}\label{subsec:eq}
Let us assume that we have a finite set $K$ of classes. We denote by $I_{\ell}^k$ the set of class $k$ users in $I_{\ell}$ and we assume that all $I_{\ell}^k$ are measurable.

Let $\sigma$ and $\hat{\sigma}$ be two Nash equilibria. We define for an $\ell \in L$, a class $k$, and an arc $a$ the quantity $$f_{\ell,a}^k=\lambda\{i\in I_\ell^k:\,a \in \sigma(i)\},$$ and $$\f_{\ell,a}^k=\lambda\{i\in I_\ell^k:\, a\in\hat\sigma(i)\}.$$ Following \citet{Mi05}, we say that the two equilibria are {\em equivalent} if not only the flow on each arc is the same but the contribution of each pair and each class to the flow on each arc is the same, i.e. $f_{\ell,a}^k = \f_{\ell,a}^k$ for any arc $a$, OD-pair $\ell$, and class $k$. Milchtaich proved that a two-terminal network has the uniqueness property if and only if every two Nash equilibria are equivalent for generically all cost functions (Theorem 5.1 in \citet{Mi05}). A property is considered {\em generic} if it holds on an open dense set. ``Open'' and ``dense'' are understood according to the following metric on the cost functions. 

Define the set $\G$ of assignments of continuous and strictly increasing cost functions $(c_a^i)_{a\in A, i\in I}$, with
$c_a^i : \mathbb{R}_+ \to \mathbb{R}_+$ such that $c_a^i=c_a^{i'}$ whenever $i$ and $i'$ belong to the same class.
 
Given a particular element of $\G$, the function $i \mapsto c_a^i(x)$ is measurable for all $a \in A$ and $x \in \mathbb{R}_+$. Every element of $\G$ has therefore a nonempty set of Nash equilibria. Note that the set $\G$ depends on the partition of the population in classes.
 We can define the distance between two elements $(c_a^i)_{a\in A, i\in I}$ and $(\tilde{c}_a^i)_{a \in A, i \in I}$ of $\G$ by $\max |c_a^i(x)-\tilde{c}_a^i(x)|$, where the maximum is taken over all $a \in A$, $i\in I$ and $x\in \mathbb{R}_+$.
 This defines a metric for $\G$.
 
\begin{theorem}
Assume that the supply graph $G$ is a cycle. Then, for any demand digraph $H$, the following assertions are equivalent:
  \begin{itemize}
   \item[]\textup{(i)} $(G,H)$ has the uniqueness property.
   \item[]\textup{(ii)} For every partition of the population into classes, there is an open dense set in $\G$ such that for any assignment of cost functions that belongs to this set, every two equilibria are equivalent.
  \end{itemize}
\end{theorem}

\begin{proof}[Proof (sketched)]
Up to slight adaptations, the proof is the same as the one of Theorem 5.1 in \citet{Mi05}. \\

If (i) does not hold, we can use the construction of Section~\ref{subsec:nonuniq} to build two distinct equilibria for an assignment in $\G$. These equilibria are such that the gap between the costs of the two routes available to any user is uniformly bounded from below by a strictly positive number. The equilibria are said to be {\em strict}.
Thus, in a  ball centered on this assignment with radius $\rho>0$ small enough, we still have two equilibria with distinct flows, which cannot be equivalent. Therefore (ii) does not hold either. \\

If (i) holds, three claims (Claims~1,~2, and~4 of \citet{Mi05}) lead to the desired conclusion, namely that (ii) holds. These three claims are now sketched. Their original proof does not need to be adapted, except for the second one, which is the only moment where the topology of the network is used. In our case the second claim gets a simpler proof.

 For an assignment in $\G$, we denote $\phi_{\ell}^k$ the number of minimal-cost routes for users in $I_\ell^k$, which is in our case $1$ or $2$. Since the uniqueness property is assumed to hold, this number is fully determined by the assignment in $\G$. Define the mean number of minimal-cost routes by $$\phi=\sum_{k\in K,\ell\in L}\lambda(I_{\ell}^k)\phi_{\ell}^k.$$

The first claim states that the map by $\phi:\G\rightarrow\mathbb{R}$ is upper semicontinuous and has finite range. 

The second claim states that for every assignment of cost functions in $\G$ that is a point of continuity of $\phi$, all Nash equilibria are equivalent. To prove this second claim, we consider two Nash equilibria assumed to be nonequivalent $\sigma$ and $\hat\sigma$. Using these two equilibria, a new one is built, $\bar\sigma$, such that for some $\ell\in L$, some class $k$ and some $\ell$-route $r_1$ we have $f_{\ell,r_1}^k>0$ and $\hat f_{\ell,r_1}^k>0$, but $\bar{f}_{\ell,r_1}^k=0$. As the two $\ell$-routes do not share any arc (Claim~\ref{claim:sum} of Section~\ref{subsec:claim}),  we have $f_{\ell,a}^k>0$, $\hat f_{\ell,a}^k>0$, and $\bar{f}_{\ell,a}^k=0$ for any $a$ in $r_1$.

The second claim is achieved by choosing any $a_1$ in $r_1$ and by adding a small value $\delta>0$ to the cost function $c_{a_1}^i$ for $i\in I_{\ell}^k$, while keeping the others unchanged. It can be checked that for $\delta$ small enough, the set of minimal-cost routes is the same as for $\delta=0$, minus the route $r_1$ for users in $I_{\ell}^k$. The map $\phi$ has therefore a discontinuity of at least $\lambda(I_\ell^k)$ at the original assignment of cost functions.

Finally, the third claim allows to conclude: in every metric space, the set of all points of continuity of a real-valued upper semicontinuous function with finite range is open and dense. 
\end{proof}

\subsection{The strong uniqueness property}\label{subsec:stronguniq}

A supply graph is said to have the {\em strong uniqueness property} if for any choice of the OD-pairs, the uniqueness property holds. In other words, $G=(V,E)$ has the strong uniqueness property if, for any digraph $H=(T,L)$ with $T\subseteq V$, the pair $(G,H)$ has the uniqueness property.

\begin{theorem}\label{thm:strong}
A graph has the strong uniqueness property if and only if no cycle is of length $3$ or more.
\end{theorem}

Alternatively, this theorem states that a graph has the strong uniqueness property if and only it is obtained by taking a forest (a graph without cycles) and by replacing some edges by parallel edges. 

Before proving this theorem, let us state a preliminary result allowing to extend the strong uniqueness property whenever one ``glues'' together two supply graphs on a vertex. This latter operation is called a {\em $1$-sum} in the usual terminology of graphs.
 
\begin{lemma}\label{lem:1sum}
The $1$-sum operation preserves the strong uniqueness property.
\end{lemma}

\begin{proof}
Let $G=(V,E)$ and $G'=(V',E')$ be two graphs, and let $H=(T,L)$ and $H'=(T',L')$ two directed graphs with $T\subseteq V$ and $T'\subseteq V'$, such that $(G,H)$ and $(G',H')$ have the uniqueness property. Assume that $(G,H)$ and $(G',H')$ have a unique common vertex $v$, i.e.  $V\cap V'=T\cap T' =\{v\}$, and define $(G'',H'')$ as the $1$-sum of them: $G''=(V\cup V', E\cup E')$ and $H''=(T\cup T', L\cup L'\cup L'')$ with $L'':=\{(u,w):\,(u,v)\in L\mbox{ and }(v,w)\in L'\}$. 

Assume that we have an equilibrium on $(G'',H'')$ for some cost functions and some partition $(I_{(o,d)})_{(o,d)\in L''}$ of the population. The restriction of this equilibrium on $(G,H)$ is an equilibrium for $(G,H)$ with the same cost functions and with a partition of the population obtained as follows.

When $o$ and $d$ are both in $H$, we keep the same $I_{(o,d)}$. Moreover, we complete this collection of subsets. For each vertex $o$ of $H$, we define $\tilde{I}_{(o,v)}$ to be the union of all $I_{(o,w)}$ with $w$ a vertex of $H'$. For each vertex $d$ of $H$, we define $\tilde{I}_{(v,d)}$ to be the union of all $I_{(w,d)}$ with $w$ a vertex of $H'$. We get the partition of the population $I$ we are looking for. The restriction of the equilibrium on $(G,H)$ is an equilibrium since for each user, the restriction of a minimum cost route of $(G'',H'')$ is a minimum cost route of $(G,H)$.

The same property holds for $(G',H')$. Therefore, if we had two equilibria inducing two distinct flows on some arc $a$ of the directed version of $G''$, we would get equilibria inducing two distinct 
flows on the arc $a$, which is in the directed version of $G$ or $G'$. It is in contradiction with the assumption on $G$ and $G'$.
\end{proof} 

\begin{proof}[Proof of Theorem~\ref{thm:strong}]
Suppose that there is a cycle $C$ of length $3$ in $G$ with vertices $u$, $v$, and $w$. Define $H$ as the digraph with arcs $(u,v)$, $(u,w)$, and $(v,w)$. The mixed graph $C+H$ is then the top left one of Figure~\ref{fig:annexe}. Corollary~\ref{cor:minor} implies that $(G,H)$ does not have the uniqueness property, and thus that $G$ does not have the strong uniqueness property.

Conversely, suppose that there is no cycle of length $3$ or more. The graph $G$ can then be obtained by successive $1$-sums of a graph made of two vertices and parallel edges. Since a graph with two vertices and parallel edges has the uniqueness property for any demand digraph (see \citet{Ko04} or \citet{Mi05}), we can conlude with Lemma~\ref{lem:1sum} that $G$ has the strong uniqueness property.
\end{proof}

\subsection{When there are only two classes}\label{subsec:two_classes}

When exhibiting multiple equilibrium flows in the proof of Theorem~\ref{thm:main}, we need to define three classes. The same remark holds for the characterization of the two-terminal networks having the uniqueness property in the article by \citet{Mi05}: all cases of non-uniqueness are built with three classes. 
We may wonder whether there are also multiple equilibrium flows with only two classes of users. The answer is yes as shown by the following examples. 
The first example is in the framework of the ring network; according to Theorem~\ref{thm:main}, such an example requires at least three OD-pairs. Since it will contain exactly three OD-pairs, it is in a sense a minimum example for ring network.
The second example involves a two-terminal network  -- $K_4$, the complete graph on four vertices -- as in \citet{BFHH09}. They used it in order  to answer a question by \citet{CoCoSt09} about the uniqueness of equilibrium in atomic player routing games. However, their cost functions do not suit our framework and we design specific ones.

\subsubsection{Multiple equilibrium flows on the ring with only two classes}

Consider the graph on top on the left of Figure~\ref{fig:annexe}. Define the two classes $1$ and $2$, with the following population measures.

$$\begin{array}{c| c cc  }
 \hline\noalign{\smallskip}
\ell\in L & (u,w) & (u,v) & (w,v)  \\ 
    \noalign{\smallskip}\hline\noalign{\smallskip}
\lambda(I_{\ell}^1) & 0 & 1.5 & 0 \\ 
\lambda(I_{\ell}^2) & 1 & 0 & 1 \\
    \noalign{\smallskip}\hline
  \end{array}$$

Cost functions are:
$$\begin{array}{c|cccccc }
 \hline\noalign{\smallskip}
\text{Arc}& (u,w) & (w,v) & (v,u) &  (w,u) & (u,v) & (v,w) \\ 
    \noalign{\smallskip}\hline\noalign{\smallskip}
\mbox{Class }1 & x & x+48 &      &  & 24x+7 & \\ 
\mbox{Class }2 & 22 x & 22 x &    &  x & x+26 & x \\
    \noalign{\smallskip}\hline
 \end{array}$$

For a given class, arcs not used in any route lead to blanks in this table.

We define the strategy profile $\sigma$ (resp. $\hat\sigma$) such that all users of the class $1$ select a negative (resp. positive) route and all users of the class $2$ select a positive (resp. negative) route. We get the following (distinct) flows.

$$\begin{array}{c| cccccc }
 \hline\noalign{\smallskip}
\text{Arc $a$} & (u,w) & (w,v) & (v,u) &  (w,u) & (u,v) & (v,w) \\ 
    \noalign{\smallskip}\hline\noalign{\smallskip}
f_a & 1  & 1 & 0 & 0 & 1.5 & 0 \\
\hat f_a & 1.5 & 1.5 & 0 & 1 & 2 & 1 \\
    \noalign{\smallskip}\hline
  \end{array}$$
  
We check that $\sigma$ is an equilibrium. 

For users in $I_{(u,v)}^1$, the cost of the positive route is $50$ and of the negative $43$. For users in $I_{(u,w)}^2$ and in $I_{(w,v)}^2$, the cost of the positive route is $22$ and of the negative $27.5$. 
No user is incitated to change its route choice. \\

We check that $\hat\sigma$ is an equilibrium. 

For users in $I_{(u,v)}^1$, the cost of the positive route is $51$ and of the negative $55$. For users in $I_{(u,w)}^2$ and in $I_{(w,v)}^2$, the cost of the positive route is $33$ and of the negative $29$. 
No user is incitated to change its route choice.

\begin{remark}
  Actually, when we specialize the construction of Section~\ref{subsec:nonuniq} to the graph on top on the left of Figure~\ref{fig:annexe}, we can merge classes $2$ and $3$ in a unique class $2$ leading to the example above. More generally, using the symmetry of the cost functions for class $2$ and class $3$ users, we can merge the two classes for any graph such that $A_{\{1,2\}}^\eps \neq \emptyset$ and $A_{\{1,3\}}^\eps \neq \emptyset$, with $\eps\in \{-,+\}$ in order to get other ring examples with two classes and multiple equilibrium flows.
\end{remark}

\subsubsection{Multiple equilibrium flows for a two-terminal network with only two classes} 

Consider the two-terminal network $K_4$ of Figure~\ref{fig:K4}.
\begin{figure}[htbp]
 \begin{center} 
\includegraphics{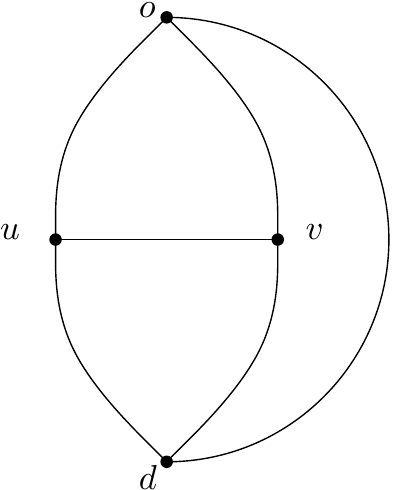}

\caption{A two-terminal network for which multiple equilibrium flows exist with only two classes}
\label{fig:K4}
 \end{center}
\end{figure}
 
 Suppose that we have two classes of users $I^1$ and $I^2$, with $\lambda(I^1)=3$ and $\lambda(I^2)=4$, with the following cost functions on each arc, where ``$\infty$'' means a prohibitively high cost function.
 $$\begin{array}{c|ccccccc}
 \hline\noalign{\smallskip}
    \text{Arc}  & (o,u) & (o,v) & (u,v) & (v,u) & (u,d) & (v,d) & (o,d) \\ 
    \noalign{\smallskip}\hline\noalign{\smallskip}
    \text{Class } 1 & x & \mbox{``$\infty$''} & x+18 & \mbox{``$\infty$''} & \mbox{``$\infty$''} & x & 7x \\
    \text{Class } 2 & 5x & x & \mbox{``$\infty$''} & \mbox{``$\infty$''} & x & 5x & x+10 \\
    \noalign{\smallskip}\hline
   \end{array}$$
Users of class $1$ have only the choice between the two routes $ouvd$ and $od$, while users of class $2$ can choose between the three routes $oud$, $ovd$, and $od$.

The strategy profile $\sigma$ is defined such that all class $1$ users select the route $ouvd$ and all class $2$ users select the route $od$.

The strategy profile $\hat\sigma$ is defined such that all class $1$ users select the route $od$, half of class $2$ users select the route $oud$, and the other half select the route $ovd$. We get the following (distinct) flows.

$$\begin{array}{c|ccccccc}
 \hline\noalign{\smallskip}
    \text{Arc $a$}  & (o,u) & (o,v) & (u,v) & (v,u) & (u,d) & (v,d) & (o,d) \\ 
        \noalign{\smallskip}\hline\noalign{\smallskip}
    f_a  & 3 & 0 & 3 & 0 & 0 & 3 & 4 \\
\hat f_a & 2 & 2 & 0 & 0 & 2 & 2 & 3 \\
    \noalign{\smallskip}\hline
\end{array}$$
   
We check that $\sigma$ is an equilibrium. 

For users of the class $1$, the cost of $ouvd$ is $27$, and the cost of $od$ is $28$. For users of the class $2$, the cost of $oud$ is $15$, the cost of $ovd$ is $15$, and the cost of $od$ is $14$. 
No user is incitated to change its route choice. \\

We check that $\hat\sigma$ is an equilibrium. 

For users of the class $1$, the cost of $ouvd$ is $22$, and the cost of $od$ is $21$. For users of the class $2$, the cost of $oud$ is $12$, the cost of $ovd$ is $12$, and the cost of $od$ is $13$. 
No user is incitated to change its route choice.

\bibliographystyle{plainnat} 
\bibliography{CongestionGames} 

\newpage
\section*{Appendix: Minimal ring graphs without the uniqueness property} \label{sec:annex}

 \begin{figure}[htbp]\begin{center}
\includegraphics[scale=0.85]{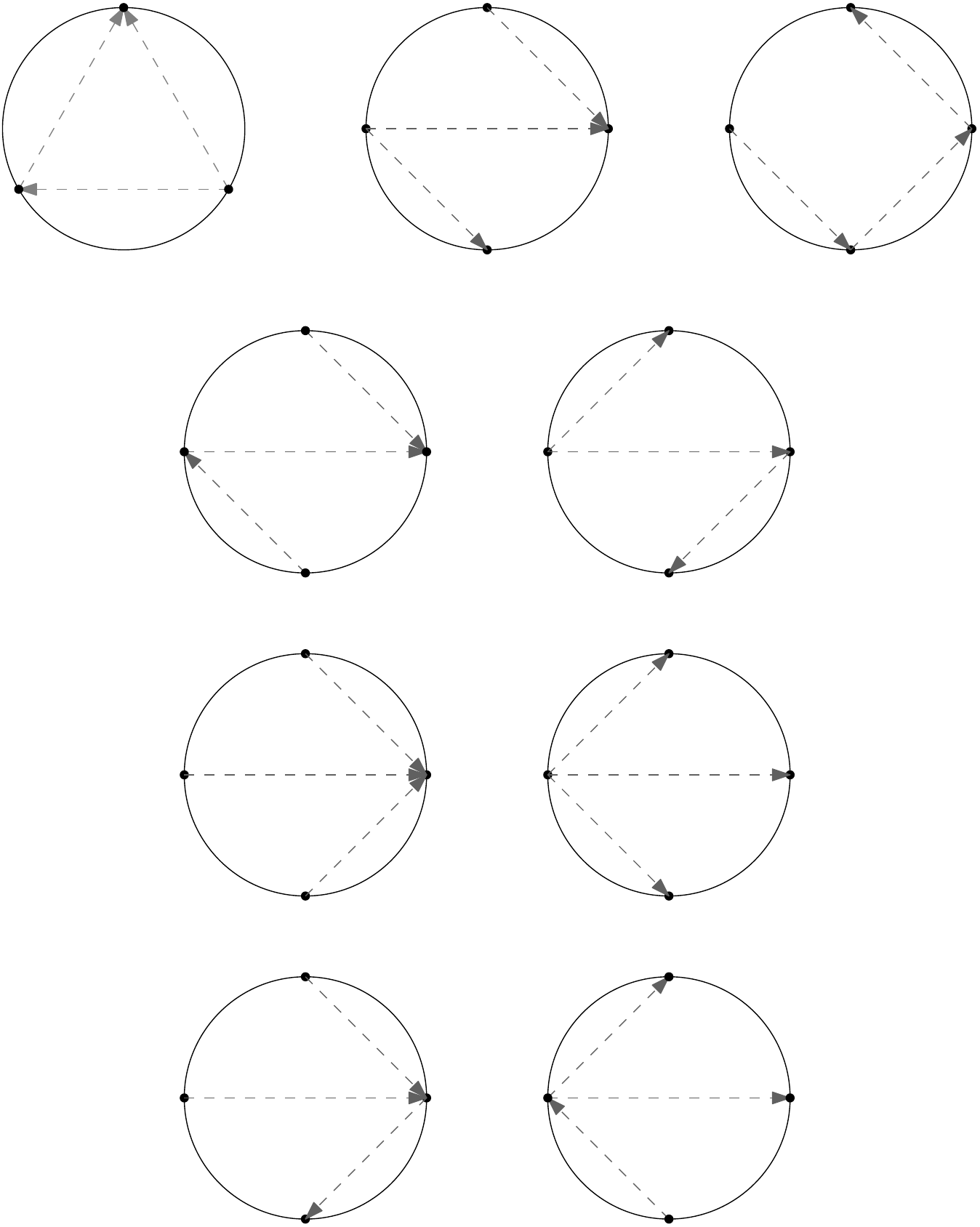}
\caption{Any ring without the uniqueness property has one of these graphs as a minor}
\label{fig:annexe}
 \end{center}\end{figure}

\end{document}